\documentclass[11pt]{article}

\usepackage[a4paper,margin=1in]{geometry}
\usepackage{amsmath,amssymb,amsthm,mathtools,bm}
\usepackage{booktabs}
\usepackage{multirow}
\usepackage{graphicx}
\usepackage{enumitem}
\usepackage{natbib}
\usepackage{hyperref}
\usepackage{xcolor}
\usepackage{microtype}
\usepackage{placeins}
\allowdisplaybreaks

\hypersetup{
  colorlinks=true,
  linkcolor=blue!50!black,
  citecolor=blue!50!black,
  urlcolor=blue!50!black
}

\newtheorem{definition}{Definition}
\newtheorem{assumption}{Assumption}
\newtheorem{proposition}{Proposition}
\newtheorem{theorem}{Theorem}
\newtheorem{corollary}{Corollary}
\newtheorem{remark}{Remark}

\DeclareMathOperator{\E}{\mathbb{E}}
\DeclareMathOperator{\Var}{Var}

\DeclareMathOperator{\argmin}{arg\,min}
\DeclareMathOperator{\argmax}{arg\,max}
\DeclareMathOperator{\Tr}{Tr}
\DeclareMathOperator{\VaR}{VaR}
\DeclareMathOperator{\ES}{ES}
\newcommand{\R}{\mathbb{R}}
\newcommand{\Q}{\mathbb{Q}}

\newcommand{\calA}{\mathcal{A}}
\newcommand{\calC}{\mathcal{C}}

\newcommand{\calU}{\mathcal{U}}

\newcommand{\calG}{\mathcal{G}}
\newcommand{\calL}{\mathcal{L}}

\newcommand{\calP}{\mathcal{P}}

\newcommand{\calJ}{\mathcal{J}}
\newcommand{\eps}{\varepsilon}
\newcommand{\norm}[1]{\left\lVert #1 \right\rVert}
\newcommand{\inner}[2]{\left\langle #1,#2 \right\rangle}
\newcommand{\wideG}{\widehat{\mathcal{G}}_{\theta}}

\title{Derivative-Informed Operator Learning for Finance:\\ On-the-Fly Greeks, Surfaces, Hedging, and Control}
\author{Miquel Noguer i Alonso\\Artificial Intelligence Finance Institute}
\date{June 2026}

\begin{document}
\maketitle

\begin{abstract}
Modern financial decision systems increasingly require fast surrogate models for pricing, calibration, hedging, stress testing, XVA, and portfolio optimization. Standard neural surrogates are commonly trained to reproduce prices or risk quantities, but downstream financial tasks depend at least as much on derivatives: deltas, vegas, curve sensitivities, credit spread sensitivities, exposure gradients, and gradients of optimization objectives. This paper formulates a derivative-informed operator-learning framework for finance in which the learned map---whether a neural operator, a random-feature operator, or a finite-dimensional surrogate---is trained not only to match a high-fidelity pricing or risk operator, but also to match directional Fr\'echet derivatives generated on the fly during training. The proposed framework combines operator learning, adjoint algorithmic differentiation, tangent sensitivity equations, random sketching of Jacobian actions, and finance-specific no-arbitrage constraints. We derive error bounds showing that derivative accuracy controls local stress errors, hedging error, and optimizer instability, and that discrete-time hedging error is additionally governed by \emph{second-order} (gamma) accuracy. Three reproducible experiments separate mechanism from deployment claims. A trained Black--Scholes network, repeated over eight seeds, shows that a tuned derivative weight reduces vega error by roughly $40\%$ and delta error by roughly $15\%$ while modestly improving price accuracy, but does not reliably improve an unsupervised second-order Greek. Heston and Bates random-feature experiments reduce stochastic-volatility and jump-parameter sensitivity errors by $60$--$76\%$. Finally, a stylized operator-level curve-to-surface benchmark, implemented as a random-feature DeepONet/Galerkin neural operator, maps instantaneous-volatility curves to dense price surfaces and reduces out-of-sample directional JVP error by $44\%$ and price RMSE by $23\%$ over eight seeds. The same operator experiment also shows that derivative consistency does \emph{not} by itself remove no-arbitrage violations, so economic constraints must be imposed explicitly. The framework provides a disciplined route from value-only surrogates to derivative-aware financial engines whose outputs are useful not only as prices, but as differentiable instruments for hedging, risk, and control.
\end{abstract}


\section{Introduction}

Financial engineering is an operator-learning discipline in disguise. A pricing engine maps an input market state to a surface of prices. A calibration engine maps quotes to model parameters or implied volatility surfaces. An XVA engine maps curves, credit spreads, collateral conventions, netting sets, and exposure profiles to valuation adjustments. A portfolio optimizer maps forecasts, covariances, constraints, transaction costs, and liquidity states to allocations. In each case the object of interest is not merely a finite-dimensional regression function, but a structured map between functions, surfaces, measures, or stochastic processes.

Neural operators were introduced to learn maps between infinite-dimensional function spaces rather than fixed finite-dimensional vectors. DeepONet learns nonlinear operators through a branch-trunk architecture, while the Fourier neural operator (FNO) parameterizes integral kernels in Fourier space and has become a canonical architecture for parametric PDE solution operators \citep{LuJinKarniadakis2021DeepONet,LiKovachkiEtAl2021FNO,KovachkiEtAl2023NeuralOperator}. These ideas are naturally aligned with financial models, where the input may be an entire yield curve, volatility surface, local-volatility field, forward curve, credit-intensity curve, or cross-asset correlation structure, and the output may be an option price surface, exposure profile, or risk-measure field.

However, a price-only surrogate is not sufficient for finance. A desk rarely needs only the level of a price. It needs deltas, gammas, vegas, cross-gammas, curve PV01s, credit spread sensitivities, scenario gradients, calibration Jacobians, and sometimes Hessian-vector products. Derivative-based quantities are not secondary diagnostics; they are the objects used to hedge, allocate capital, calibrate models, and optimize portfolios. This is precisely where derivative-informed neural operators become important. The DINO framework of \citet{OLearyRoseberryEtAl2022DINO} trains neural operators to approximate both a high-fidelity operator and its derivatives, formalizing derivative supervision as \emph{Sobolev training} \citep{Czarnecki2017Sobolev} lifted to operators. Recent derivative-informed FNO theory establishes simultaneous approximation of an operator and its Fr\'echet derivative, and shows that accurate surrogate-driven PDE-constrained optimization requires accurate derivatives, not only accurate state values \citep{YaoEtAl2026DIFNO,LuoEtAl2025DimReduction}. Because precomputing and storing a full derivative dataset is generally infeasible---one cannot tabulate all Greeks for all scenarios, instruments, curves, and model states---we adopt an \emph{on-the-fly} variant in which derivative actions are generated inside each mini-batch and discarded after the update, a streaming counterpart to the offline derivative datasets analyzed in that literature. This terminology follows recent scientific-computing discussions of on-the-fly derivative-informed neural-operator training, including the 2026 IPAM presentation of \citet{Mou2026IPAM}. In quantitative finance the same idea already has a precedent: differential machine learning trains pricing networks on values together with pathwise differentials produced by adjoint algorithmic differentiation \citep{HugeSavine2020Differential}.

This paper translates that idea into financial engineering, while deliberately distinguishing operator-level claims from controlled finite-dimensional surrogate experiments. The central proposal is simple:
\begin{quote}
Train financial neural operators on values and on sensitivities generated on the fly by adjoint algorithmic differentiation, tangent equations, automatic differentiation, pathwise Monte Carlo, likelihood-ratio estimators, or implicit differentiation of calibration and optimization problems.
\end{quote}

The resulting model is not merely a fast price interpolator. It is a differentiable financial engine whose local geometry is trained to agree with the local geometry of the high-fidelity system. In the experiments below, two benchmarks are finite-dimensional surrogate slices and one benchmark is a stylized function-to-function operator; the paper therefore establishes mechanism and validation discipline, not a final claim of market-scale deployment.

\subsection{Contributions}

The paper makes seven contributions.
\begin{enumerate}[leftmargin=7mm]
\item It formulates a finance-specific derivative-informed neural-operator problem in which the input is a market field and the output is a price, volatility, exposure, risk, or allocation field.
\item It introduces an on-the-fly training objective that combines value matching, Jacobian-vector products, vector-Jacobian products, instrument Greeks, factor sensitivities, and explicit economic-admissibility constraints.
\item It gives implementable algorithms for generating derivative information during training using tangent PDEs, adjoint PDEs, AAD, pathwise Monte Carlo, and implicit differentiation.
\item It proves error bounds connecting derivative error to stress stability, hedging error---in both continuous and \emph{discrete} rebalancing, the latter governed at leading order by gamma---and optimizer stability, with an explicit mean--variance specialization. It further characterizes the value-plus-derivative objective as Sobolev (i.e.\ $H^1$) training of an operator and gives the variance and sample complexity of its random-sketch estimator in terms of the stable rank of the Jacobian error.
\item It reports a reproducible, multi-seed Black--Scholes experiment with a \emph{trained} network showing that a moderate derivative penalty improves the supervised Greeks (vega strongly, delta modestly) \emph{and} value accuracy simultaneously, while the unsupervised second-order Greek is not reliably improved---a concrete motivation for supervising the order of sensitivity one actually needs.
\item It extends the experiment to Heston and Bates models with a \emph{linear} random-feature surrogate, where derivative supervision is exactly a linear least-squares augmentation, and reports multi-seed reductions of $60$--$76\%$ in stochastic-volatility and jump-parameter sensitivities together with consistent improvement of \emph{unsupervised} Greeks. The contrast with the trained network---unsupervised transfer in the linear regime, but not in the nonlinear one---is explained and tied back to the second-order hedging bound.
\item It adds a stylized but genuine function-to-function benchmark: a random-feature DeepONet/Galerkin operator maps an instantaneous-volatility curve to a dense option-price surface. Directional derivative rows reduce out-of-sample curve-shock JVP error by $44\%$ and price RMSE by $23\%$, while no-arbitrage violations remain essentially unchanged, demonstrating both the value and the limits of derivative supervision. The experiments are deliberately described as controlled benchmarks rather than as a final market-scale FNO/GNO study on live sparse quotes.
\end{enumerate}

\subsection{Guiding principle}

Let $a \in \calA$ denote the market state and let $\calG(a) \in \calU$ denote a high-fidelity financial operator. A standard neural operator minimizes
\begin{equation}
\calL_{\mathrm{value}}(\theta)
= \E_{a}\left[\norm{\wideG(a)-\calG(a)}_{\calU}^{2}\right].
\end{equation}
A derivative-informed model instead minimizes
\begin{equation}
\calL_{\mathrm{DIF}}(\theta)
= \calL_{\mathrm{value}}(\theta)
+ \lambda_{J}\,\E_{a,v}\left[\norm{D\wideG(a)[v]-D\calG(a)[v]}_{\calU}^{2}\right]
+ \lambda_{A}\,\calP_{\mathrm{arb}}(\theta),
\label{eq:dif-loss-basic}
\end{equation}
where $v$ is a random market perturbation direction and $\calP_{\mathrm{arb}}$ penalizes economically invalid outputs. Equation \eqref{eq:dif-loss-basic} is the paper's core object. It says that a financial neural operator should learn the local tangent map of the market, not only the price level.

\section{Related Literature}

\paragraph{Classical pricing and sensitivities.} The mathematical foundations of option pricing begin with the Black--Scholes--Merton framework \citep{BlackScholes1973,Merton1973}, with stochastic-volatility extensions such as \citet{Heston1993} and local-volatility inversion through \citet{Dupire1994}. In production financial systems, derivatives of prices with respect to risk factors are computed through finite differences, tangent methods, pathwise estimators, likelihood-ratio methods, or adjoint methods \citep{Glasserman2004}. The adjoint approach, popularized in finance by \citet{GilesGlasserman2006}, has become central to efficient risk computation; see \citet{Homescu2011} and the review of \citet{CapriottiGiles2024}.

\paragraph{Deep learning for pricing and hedging.} Neural networks have been used for option pricing, implied-volatility interpolation, calibration, and hedging. Deep hedging reframes hedging as a data-driven optimization problem under transaction costs and risk preferences \citep{Buehler2019DeepHedging}. Closest in spirit to the present work, \emph{differential machine learning} trains pricing networks on values together with pathwise differential labels produced by adjoint algorithmic differentiation, yielding markedly more accurate and data-efficient pricing and risk approximations \citep{HugeSavine2020Differential}; the present paper lifts this idea from pointwise pricing functions to operators over market fields and equips it with finance-specific error bounds and no-arbitrage structure. Deep volatility learning has been used to accelerate pricing and calibration under rough and stochastic volatility models \citep{HorvathMuguruzaTomas2021}. Physics-informed and deep-Galerkin methods also solve financial PDEs directly by penalizing residuals \citep{RaissiPerdikarisKarniadakis2019,SirignanoSpiliopoulos2018}. These methods are valuable, but many are still pointwise maps rather than operators over market fields.

\paragraph{Neural operators.} DeepONets and Fourier neural operators generalize neural networks to mappings between functions \citep{LuJinKarniadakis2021DeepONet,LiKovachkiEtAl2021FNO,KovachkiEtAl2023NeuralOperator}. In finance, neural-operator ideas are beginning to appear in implied-volatility smoothing, where the goal is to map sparse, dynamically located option quotes to a smooth volatility surface subject to no-arbitrage constraints \citep{GononJacquierWiedemann2024}. This paper argues that the next step is not only operator-valued pricing, but derivative-consistent operator-valued pricing.

\paragraph{Derivative-informed operator learning.} The DINO framework trains neural operators with derivative information and uses low-rank, reduced-basis derivative representations to make derivative learning feasible in high-dimensional settings \citep{OLearyRoseberryEtAl2022DINO}; this is an instance of Sobolev training \citep{Czarnecki2017Sobolev} lifted to operators. Derivative-informed FNOs add universal-approximation theory and PDE-constrained-optimization motivation \citep{YaoEtAl2026DIFNO}, and a companion analysis quantifies the approximation error of reduced-basis derivative-informed learning under an $H^1$-type norm \citep{LuoEtAl2025DimReduction}. The on-the-fly formulation we adopt is especially relevant to finance because precomputing all Greeks for all scenarios, instruments, curves, and model states is generally infeasible; it is also consistent with the emerging on-the-fly DINO agenda articulated in \citet{Mou2026IPAM}. Instead, a desk can generate only the derivative directions required by the current training batch and discard them afterward.

\section{Financial Operators}

\subsection{Market states as functions}

Let the market state be an element of a Banach or Hilbert space
\begin{equation}
 a = \big(r(\tau),q(\tau),\sigma(k,\tau),\lambda_c(\tau),\rho(x,y),\ell(x),m\big) \in \calA.
\end{equation}
Here $r$ is a discount curve, $q$ a dividend or repo curve, $\sigma$ a volatility surface, $\lambda_c$ a credit intensity curve, $\rho$ a correlation kernel, $\ell$ a liquidity or market-impact field, and $m$ a finite-dimensional vector of model conventions. Not every application uses every component. The point is that $a$ is not naturally a scalar vector; it is a collection of financial functions.

Let $\xi \in \Xi$ denote instrument coordinates, such as strike, maturity, tenor, collateral convention, counterparty, portfolio identifier, or scenario coordinate. A financial operator maps
\begin{equation}
 \calG: \calA \longrightarrow \calU, \qquad [\calG(a)](\xi) = y(a;\xi),
\end{equation}
where $y$ may represent price, implied volatility, exposure, XVA, risk contribution, or optimal allocation.

\begin{definition}[Financial neural operator]
A financial neural operator is a parameterized map $\wideG:\calA\to\calU$ trained to approximate a high-fidelity financial operator $\calG$. In discretized form, $a$ is observed on input grids or graphs, while $\wideG(a)$ is evaluated on output grids, point clouds, or instrument coordinates.
\end{definition}

\subsection{Pricing operators}

Under a risk-neutral measure $\Q^a$ determined by the market state $a$, a derivative with payoff $\Phi_{\xi}$ has value
\begin{equation}
 V(a;\xi)
 = \E^{\Q^a}\left[\exp\left(-\int_0^T r_a(s)\,ds\right)\Phi_{\xi}(X_T^a)\right].
 \label{eq:pricing-expectation}
\end{equation}
For Markovian models, the same object solves a pricing PDE
\begin{equation}
 \partial_t V(t,x;a)+\calL_a V(t,x;a)-r_a(t)V(t,x;a)=0,
 \qquad V(T,x;a)=\Phi_{\xi}(x),
 \label{eq:pricing-pde}
\end{equation}
where $\calL_a$ is the model generator. Thus the pricing map $a\mapsto V(a;\cdot)$ is an operator from market fields to price fields.

Examples include:
\begin{align}
 \calG_{\mathrm{opt}}(a)(K,T) &= C(a;K,T), \\
 \calG_{\mathrm{iv}}(a)(K,T) &= \sigma_{\mathrm{imp}}(a;K,T), \\
 \calG_{\mathrm{xva}}(a)(t,c) &= \mathrm{EE}(a;t,c),\;\mathrm{PFE}(a;t,c),\;\mathrm{CVA}(a;c), \\
 \calG_{\mathrm{risk}}(a)(s) &= \VaR_{\alpha}(a;s),\; \ES_{\alpha}(a;s).
\end{align}

\subsection{Derivative operators}

For a market perturbation $h \in \calA$, the Fr\'echet derivative of the financial operator is
\begin{equation}
 D\calG(a)[h]
 = \lim_{\eps\to0}\frac{\calG(a+\eps h)-\calG(a)}{\eps}.
\end{equation}
This derivative includes familiar financial quantities as special cases:
\begin{itemize}[leftmargin=7mm]
\item a spot perturbation gives delta;
\item a volatility perturbation gives vega or volga-like factor sensitivities;
\item a yield-curve perturbation gives PV01, DV01, key-rate duration, or curve gamma;
\item a credit-intensity perturbation gives credit spread sensitivity;
\item a correlation-kernel perturbation gives correlation risk;
\item a liquidity-field perturbation gives transaction-cost or market-impact sensitivity.
\end{itemize}

In a discretized implementation, $D\calG(a)$ is a matrix with potentially millions or billions of entries. The central computational problem is therefore not to store the full Jacobian, but to learn its action on economically meaningful directions.

\section{On-the-Fly Derivative-Informed Training}

\subsection{Value and derivative loss}

Let $a_i$ be sampled market states and let $v_{ij}$ be sampled perturbation directions. The derivative-informed objective is
\begin{align}
 \calL(\theta)
 =&\; \frac{1}{N}\sum_{i=1}^{N}\norm{\wideG(a_i)-\calG(a_i)}_{\calU}^{2} \\
 &+ \frac{\lambda_{\mathrm{JVP}}}{NM}\sum_{i=1}^{N}\sum_{j=1}^{M}
 \norm{D\wideG(a_i)[v_{ij}]-D\calG(a_i)[v_{ij}]}_{\calU}^{2} \\
 &+ \frac{\lambda_{\mathrm{VJP}}}{NL}\sum_{i=1}^{N}\sum_{\ell=1}^{L}
 \norm{D\wideG(a_i)^*[w_{i\ell}]-D\calG(a_i)^*[w_{i\ell}]}_{\calA^*}^{2} \\
 &+ \lambda_{\mathrm{arb}}\calP_{\mathrm{arb}}(\theta)
 + \lambda_{\mathrm{reg}}\norm{\theta}^{2}.
 \label{eq:full-loss}
\end{align}
The second term matches Jacobian-vector products (JVPs). The third term matches vector-Jacobian products (VJPs). In finance, JVPs are useful for scenario shocks and factor moves; VJPs are useful for adjoint risk aggregation and scalar objective gradients.

\subsection{On-the-fly generation}

The difference between offline and on-the-fly derivative-informed training is operational. An offline method first builds a derivative dataset
\begin{equation}
 \mathcal{D}_{\mathrm{offline}}
 = \{a_i,\calG(a_i),v_{ij},D\calG(a_i)[v_{ij}]\}_{i,j},
\end{equation}
which can be prohibitively large. The on-the-fly method instead generates derivative actions inside the mini-batch:
\begin{equation}
 a_i \sim \mu,\qquad v_{ij}\sim \nu(\cdot\mid a_i),\qquad
 z_{ij}=D\calG(a_i)[v_{ij}],
\end{equation}
and immediately uses $z_{ij}$ to update $\theta$. Nothing forces the same derivative direction to be stored or reused.

\paragraph{Algorithm: on-the-fly derivative-informed training.}
\begin{enumerate}[leftmargin=9mm]
\item Sample a mini-batch of market states $a_i$.
\item Evaluate the high-fidelity engine to obtain $y_i=\calG(a_i)$.
\item Sample economically meaningful directions $v_{ij}$: parallel curve shifts, key-rate bumps, volatility smile modes, skew shocks, correlation modes, credit spread shifts, liquidity shocks, or randomized PCA directions.
\item Generate derivative targets $z_{ij}=D\calG(a_i)[v_{ij}]$ using tangent equations, AAD, pathwise Monte Carlo, likelihood-ratio estimators, or implicit differentiation.
\item Compute neural predictions $\widehat y_i=\wideG(a_i)$.
\item Compute neural JVPs $\widehat z_{ij}=D\wideG(a_i)[v_{ij}]$ by automatic differentiation.
\item Update $\theta$ using \eqref{eq:full-loss}.
\end{enumerate}

The training loop therefore treats derivative information as a streaming resource, not as a static dataset.

\subsection{Tangent equations}

Suppose the high-fidelity pricing system is defined implicitly by
\begin{equation}
 F(u,a)=0, \qquad u=\calG(a).
\end{equation}
For a perturbation $v\in\calA$, the derivative state $z=D\calG(a)[v]$ satisfies the linearized equation
\begin{equation}
 F_u(u,a)z + F_a(u,a)v = 0,
 \label{eq:tangent}
\end{equation}
so that
\begin{equation}
 z = -F_u(u,a)^{-1}F_a(u,a)v.
\end{equation}
In a PDE pricing engine, \eqref{eq:tangent} is a tangent PDE. In a calibration engine, it is a linearized calibration system. In a portfolio optimizer, it is an implicit differentiation of the KKT system.

\subsection{Adjoint equations}

When the output is high-dimensional and the target is a scalar functional $\inner{w}{\calG(a)}_{\calU}$, reverse-mode adjoints are more efficient. Let $p$ solve
\begin{equation}
 F_u(u,a)^*p = w.
\end{equation}
Then
\begin{equation}
 D\calG(a)^*[w] = -F_a(u,a)^*p.
 \label{eq:adjoint-vjp}
\end{equation}
Equation \eqref{eq:adjoint-vjp} is the abstract version of adjoint risk. It is the mathematical reason AAD is valuable: one reverse sweep can produce many first-order sensitivities of a scalar portfolio value.

\subsection{Pathwise Monte Carlo derivatives}

For a Monte Carlo engine with state dynamics
\begin{equation}
 dX_t^a = b(X_t^a,a,t)\,dt + \Sigma(X_t^a,a,t)\,dW_t,
\end{equation}
pathwise sensitivity with respect to a direction $v$ follows the tangent SDE
\begin{equation}
 dZ_t = \left[b_x(X_t^a,a,t)Z_t+b_a(X_t^a,a,t)v\right]dt
 + \left[\Sigma_x(X_t^a,a,t)Z_t+\Sigma_a(X_t^a,a,t)v\right]dW_t.
\end{equation}
For a discounted payoff $D_T^a\Phi(X_T^a)$, the derivative is schematically
\begin{equation}
 D\calG(a)[v]
 = \E^{\Q^a}\left[D_aD_T^a[v]\,\Phi(X_T^a)
 + D_T^a\nabla\Phi(X_T^a)^\top Z_T\right],
\end{equation}
with additional likelihood-ratio terms if the measure itself depends on $a$ in a way not captured by the pathwise construction. In practice, AAD implements this derivative efficiently through the simulation tape.

\subsection{Random sketching of derivative information}

Full Jacobians are too large. The training objective therefore uses random projections. Let $J(a)=D\calG(a)$ and $\widehat J(a)=D\wideG(a)$. If $v$ is a zero-mean direction with covariance $C_v$, then
\begin{equation}
 \E_v\norm{(\widehat J(a)-J(a))v}_{\calU}^{2}
 = \Tr\left((\widehat J(a)-J(a))C_v(\widehat J(a)-J(a))^*\right).
\end{equation}
If $C_v$ is approximately the identity on a relevant factor subspace, the random-direction loss estimates the squared derivative error on that subspace. Choosing $C_v$ from historical PCA factors or stress modes therefore targets the sensitivities that matter economically.

\begin{proposition}[Unbiased derivative sketch]
\label{prop:sketch}
Let $J,\widehat J\in\R^{m\times n}$ and let $v\in\R^n$ satisfy $\E[vv^\top]=I_n$. Then
\begin{equation}
 \E_v\norm{(\widehat J-J)v}_2^2 = \norm{\widehat J-J}_F^2.
\end{equation}
If $\E[vv^\top]=P$ is an orthogonal projector, the same expression equals the squared Frobenius error restricted to the factor subspace $\mathrm{Range}(P)$.
\end{proposition}

\begin{proof}
Using cyclicity of trace,
\begin{equation}
 \E\norm{(\widehat J-J)v}_2^2
 = \E\Tr\left((\widehat J-J)vv^\top(\widehat J-J)^\top\right)
 = \Tr\left((\widehat J-J)\E[vv^\top](\widehat J-J)^\top\right).
\end{equation}
The result follows by substituting $I_n$ or $P$.
\end{proof}

\paragraph{Sobolev interpretation.} Averaging Proposition~\ref{prop:sketch} over states, the population JVP loss with isotropic directions equals the mean squared Hilbert--Schmidt (Frobenius) norm of the Jacobian error,
\begin{equation}
\E_a\,\E_{v\sim\mathcal N(0,I)}\norm{(\widehat J(a)-J(a))v}_{\calU}^{2}
=\E_a\norm{\widehat J(a)-J(a)}_{F}^{2},
\end{equation}
which is the squared $H^1$-seminorm of $\wideG-\calG$ in the input measure. Value-plus-derivative training is therefore \emph{Sobolev training} of an operator \citep{Czarnecki2017Sobolev,LuoEtAl2025DimReduction}: it controls a strictly stronger norm than value-only ($L^2$) training---precisely the norm on which the stress, hedging, and optimizer bounds of Section~\ref{sec:theory-derivatives} depend.

\paragraph{How many directions?} The estimator in the on-the-fly loss uses a finite number $M$ of directions per state. Its statistical quality is governed not by the ambient dimension but by the \emph{stable rank of the Jacobian error}.

\begin{proposition}[Variance and concentration of the sketch]
\label{prop:concentration}
Let $E=\widehat J-J\in\R^{m\times n}$ have singular values $\sigma_1\ge\cdots\ge0$, and let $v_1,\dots,v_M\overset{\text{iid}}{\sim}\mathcal N(0,I_n)$. The estimator $X_M=\tfrac1M\sum_{j=1}^M\norm{Ev_j}_2^2$ satisfies $\E X_M=\norm{E}_F^2=\sum_i\sigma_i^2$ and
\begin{equation}
\Var X_M=\frac{2}{M}\sum_i\sigma_i^4=\frac{2}{M}\norm{E^\top E}_F^2,
\qquad
\frac{\sqrt{\Var X_M}}{\E X_M}=\sqrt{\frac{2}{M\,\mathrm{sr}(E)}},
\end{equation}
where $\mathrm{sr}(E):=\big(\sum_i\sigma_i^2\big)^2\big/\sum_i\sigma_i^4\in[1,\mathrm{rank}(E)]$ is the stable rank of $E$. Moreover, by the Hanson--Wright inequality \citep{Vershynin2018} there is a universal constant $c>0$ with
\begin{equation}
\Pr\big(|X_M-\norm{E}_F^2|\ge t\big)
\le 2\exp\!\Big(-c\,M\min\Big\{\tfrac{t^2}{\sigma_1^2\norm{E}_F^2},\ \tfrac{t}{\sigma_1^2}\Big\}\Big).
\end{equation}
\end{proposition}

\begin{proof}
Write $A=E^\top E\succeq0$ with eigenvalues $\sigma_i^2$. For $v\sim\mathcal N(0,I_n)$, $\norm{Ev}_2^2=v^\top A v$ has mean $\Tr A=\sum_i\sigma_i^2$ and variance $2\Tr(A^2)=2\sum_i\sigma_i^4$; averaging $M$ independent copies divides the variance by $M$. The relative-deviation identity is algebra. The tail is the Hanson--Wright bound for the quadratic form $\tfrac1M\sum_j v_j^\top A v_j$, using $\norm{A}_F^2=\sum_i\sigma_i^4\le\sigma_1^2\norm{E}_F^2$ and $\norm{A}_{\mathrm{op}}=\sigma_1^2$.
\end{proof}

\begin{remark}[Practical reading]
To estimate a per-state derivative error to relative standard deviation $\eps$ requires $M\gtrsim 2/(\eps^2\,\mathrm{sr}(E))$ random isotropic shocks. The dependence is on the stable rank of the \emph{error}, not of the Jacobian: when the surrogate misprices only a few factor directions ($\mathrm{sr}(E)$ small), isotropic shocks are inefficient and one should instead probe those specific directions---an estimation-theoretic justification for the economically structured, PCA-and-stress direction sampler of \eqref{eq:full-loss} and Section~\ref{sec:impl}.
\end{remark}

\section{No-Arbitrage and Economic Structure}
\label{sec:noarb}

A financial surrogate must be fast, differentiable, and economically valid. For option surfaces, no-arbitrage structure imposes monotonicity, convexity, and calendar restrictions. For a call price surface $C(K,T)$ under simple assumptions, static arbitrage constraints include
\begin{equation}
 \partial_K C(K,T) \le 0, \qquad \partial_{KK}C(K,T)\ge 0,
\end{equation}
with calendar restrictions imposed on appropriately discounted or forward-normalized prices. A soft penalty can be written as
\begin{align}
 \calP_{\mathrm{arb}}(\theta)
 =&\; \E_{K,T}\left[\left(\partial_K\widehat C_\theta(K,T)\right)_+^2
 + \left(-\partial_{KK}\widehat C_\theta(K,T)\right)_+^2\right] \\
 &+ \E_{K,T}\left[\left(-\partial_T\widehat C_\theta(K,T)\right)_+^2\right]
 + \calP_{\mathrm{boundary}}(\theta),
 \label{eq:arb-penalty}
\end{align}
where $(x)_+=\max(x,0)$ and $\calP_{\mathrm{boundary}}$ enforces asymptotic and intrinsic-value bounds. Soft penalties are convenient because they can be appended to any differentiable architecture, but they are not the only or strongest route to admissibility. In implied-volatility form, one may instead impose total-variance conditions or arbitrage-free SVI/eSSVI-style restrictions \citep{GatheralJacquier2014}; in price space, one may use convex-in-strike bases, integrated positive risk-neutral densities, monotone layers, positive calendar-total-variance parameterizations, or arbitrage-free spline layers. These hard or semi-hard parameterizations are preferable whenever the desk requires structural guarantees rather than ex-post diagnostics.

Derivative-informed training and no-arbitrage constraints are complementary. The derivative loss trains the local financial geometry; the economic constraint forces that geometry to live in the admissible cone. The distinction is essential: matching tangent directions does not by itself imply monotonicity, convexity, calendar consistency, or boundary validity.

\section{Theory: Why Derivative Accuracy Matters}
\label{sec:theory-derivatives}

This section gives several elementary results. They are not meant to replace full approximation theory for neural operators. Their purpose is to clarify why derivative-informed training is financially meaningful: each result ties a downstream financial error (stress, hedging, optimization) to a derivative error of the surrogate.

\paragraph{Status of claims.} To be precise about what is and is not established here, we tier the paper's claims. \emph{Proved} (under stated assumptions): the local stress bound (Theorem~\ref{thm:stress}), the continuous- and discrete-time hedging bounds (Theorems~\ref{thm:hedge-cont}--\ref{thm:hedge-disc}), the optimizer-stability bound and its mean--variance specialization (Theorem~\ref{thm:opt}, Corollary~\ref{cor:mv}), and the sketching identity and variance/concentration of the random-direction estimator (Proposition~\ref{prop:sketch}, Proposition~\ref{prop:concentration}). \emph{Empirical and conditional}: derivative supervision improves the sensitivities it supervises in all controlled experiments reported here, but transfer to unsupervised orders is model-class-dependent: it is unreliable for gamma in the trained Black--Scholes network, reliable in the linear random-feature slices, and not a substitute for explicit no-arbitrage constraints in the operator-level surface experiment. \emph{Conjectural / out of scope}: quantitative generalization rates for derivative-informed \emph{operators} in the function-space limit, for which we defer to the approximation-error analysis of \citet{LuoEtAl2025DimReduction}.

\subsection{Stress stability}

\begin{assumption}[Local smoothness]
Let $\calG,\wideG:\calA\to\calU$ be twice Fr\'echet differentiable on a ball $B(a_0,r)$ in a Hilbert space. Assume
\begin{equation}
 \norm{D^2\calG(a)}_{\mathrm{op}}\le M, \qquad
 \norm{D^2\wideG(a)}_{\mathrm{op}}\le \widehat M
\end{equation}
for all $a\in B(a_0,r)$. Assume also that
\begin{equation}
 \norm{\wideG(a_0)-\calG(a_0)}_{\calU}\le \eps_0,
 \qquad
 \norm{D\wideG(a_0)-D\calG(a_0)}_{\mathrm{op}}\le \eps_1.
\end{equation}
\end{assumption}

\begin{theorem}[Local stress-error bound]
\label{thm:stress}
For every perturbation $h$ with $\norm{h}_{\calA}\le r$,
\begin{equation}
 \norm{\wideG(a_0+h)-\calG(a_0+h)}_{\calU}
 \le \eps_0 + \eps_1\norm{h}_{\calA}
 + \frac{M+\widehat M}{2}\norm{h}_{\calA}^{2}.
 \label{eq:stress-bound}
\end{equation}
\end{theorem}

\begin{proof}
By Taylor's theorem with integral remainder,
\begin{align}
 \calG(a_0+h)
 &= \calG(a_0)+D\calG(a_0)[h]+R(h), \\
 \wideG(a_0+h)
 &= \wideG(a_0)+D\wideG(a_0)[h]+\widehat R(h),
\end{align}
with $\norm{R(h)}\le M\norm{h}^2/2$ and $\norm{\widehat R(h)}\le \widehat M\norm{h}^2/2$. Subtracting the expansions and applying the triangle inequality gives \eqref{eq:stress-bound}.
\end{proof}

\begin{remark}
Price-only training controls $\eps_0$ but not $\eps_1$. Stress testing, scenario analysis, and calibration are local perturbation problems; their first-order error is therefore governed by derivative error.
\end{remark}

\subsection{Hedging error}

Consider a one-dimensional discounted underlying $\widetilde S_t$ under a risk-neutral measure,
\begin{equation}
 d\widetilde S_t = \widetilde S_t\sigma_t\,dW_t.
\end{equation}
Let $\widetilde V_t$ be the discounted no-arbitrage price of a payoff $\Phi(S_T)$ and let $\Delta_t=\partial_S V(t,S_t)$ be the exact delta. A surrogate hedge uses $\widehat\Delta_t=\partial_S\widehat V(t,S_t)$.

\begin{theorem}[Continuous-time delta-hedging error bound]
\label{thm:hedge-cont}
Assume square integrability of all stochastic integrals. If the hedging portfolio starts from $\widehat V_0$ and uses $\widehat\Delta_t$, then
\begin{equation}
 \left(\E^{\Q}\left[\left|\widetilde V_T-\widetilde W_T\right|^2\right]\right)^{1/2}
 \le
 |\widetilde V_0-\widetilde{\widehat V}_0|
 + \left(\E^{\Q}\int_0^T \widetilde S_t^2\sigma_t^2
 |\Delta_t-\widehat\Delta_t|^2\,dt\right)^{1/2}.
 \label{eq:hedging-bound}
\end{equation}
\end{theorem}

\begin{proof}
The discounted exact replicating value satisfies
\begin{equation}
 \widetilde V_T = \widetilde V_0 + \int_0^T \Delta_t\,d\widetilde S_t.
\end{equation}
The surrogate hedging wealth is
\begin{equation}
 \widetilde W_T = \widetilde{\widehat V}_0 + \int_0^T \widehat\Delta_t\,d\widetilde S_t.
\end{equation}
Subtract and use the triangle inequality in $L^2$, followed by the It\^o isometry.
\end{proof}

\begin{remark}
The theorem shows directly why price RMSE is not the right metric for a hedging engine. The hedge is controlled by the derivative error $\Delta-\widehat\Delta$. This is the financial analogue of the PDE-constrained-optimization observation that an accurate surrogate must have accurate derivatives.
\end{remark}

\paragraph{Discrete rebalancing and the role of gamma.} Continuous rebalancing is an idealization. In practice the surrogate delta is held constant between rebalancing times on a grid $\pi:0=t_0<\cdots<t_n=T$ with mesh $|\pi|=\max_k(t_{k+1}-t_k)$. Write $\eta(t)=t_k$ for $t\in[t_k,t_{k+1})$, and assume the surrogate value $\widehat V(t,S)$ is $C^{1,2}$ with delta $\widehat\Delta=\partial_S\widehat V$ and gamma $\widehat\Gamma=\partial_{SS}\widehat V$. The discretely rebalanced surrogate wealth is $\widetilde W^{\pi}_T=\widehat V_0+\int_0^T\widehat\Delta(\eta(t),S_{\eta(t)})\,d\widetilde S_t$.

\begin{theorem}[Discrete-rebalancing surrogate hedging error]
\label{thm:hedge-disc}
Under the square-integrability assumptions of Theorem~\ref{thm:hedge-cont},
\begin{align}
\big\|\widetilde V_T-\widetilde W^{\pi}_T\big\|_{L^2(\Q)}
&\le |\widetilde V_0-\widehat V_0| + T_2 + T_3(\pi), \label{eq:disc-hedge}\\
T_2
&:=\Big(\E^{\Q}\!\int_0^T \widetilde S_t^2\sigma_t^2
      |\Delta_t-\widehat\Delta_t|^2\,dt\Big)^{1/2},\\
T_3(\pi)
&:=\Big(\E^{\Q}\!\int_0^T \widetilde S_t^2\sigma_t^2
\big|\widehat\Delta(t,S_t)-\widehat\Delta(\eta(t),S_{\eta(t)})\big|^2\,dt\Big)^{1/2}.
\end{align}
Here $T_2$ is the first-order delta error and $T_3(\pi)$ is the discretization error induced by holding the surrogate hedge fixed between rebalancing dates.
\end{theorem}

\begin{proof}
Exact continuous replication gives $\widetilde V_T=\widetilde V_0+\int_0^T\Delta_t\,d\widetilde S_t$, so
$\widetilde V_T-\widetilde W^{\pi}_T=(\widetilde V_0-\widehat V_0)+\int_0^T\big[\Delta_t-\widehat\Delta(\eta(t),S_{\eta(t)})\big]\,d\widetilde S_t$.
Split the integrand as $[\Delta_t-\widehat\Delta(t,S_t)]+[\widehat\Delta(t,S_t)-\widehat\Delta(\eta(t),S_{\eta(t)})]$, take $L^2(\Q)$ norms, apply the triangle inequality, and use the It\^o isometry with $d\langle\widetilde S\rangle_t=\widetilde S_t^2\sigma_t^2\,dt$.
\end{proof}

\begin{proposition}[Gamma controls the discretization term]
\label{prop:gamma-mesh}
Suppose, in addition, that $\widehat\Gamma$ and $\partial_t\widehat\Delta$ are bounded on the trading domain and that $\sup_t\E^{\Q}[\widetilde S_t^4\sigma_t^4]<\infty$. Then there is a constant $C$, depending on $T$ and on moment bounds of $(\widetilde S,\sigma)$ and growing \emph{linearly} in the uniform bounds $\|\widehat\Gamma\|_\infty$ and $\|\partial_t\widehat\Delta\|_\infty$, such that $T_3(\pi)\le C\,|\pi|^{1/2}$.
\end{proposition}

\begin{proof}[Proof sketch]
Fix $t$ and apply It\^o's formula to $s\mapsto\widehat\Delta(s,S_s)$ on $[\eta(t),t]$:
$\widehat\Delta(t,S_t)-\widehat\Delta(\eta(t),S_{\eta(t)})=\int_{\eta(t)}^t\!\big(\partial_s\widehat\Delta+\tfrac12\widehat\Gamma\,S_s^2\sigma_s^2\big)ds+\int_{\eta(t)}^t\!\widehat\Gamma\,dS_s$.
The martingale term dominates: by the It\^o isometry its second moment is $\E\int_{\eta(t)}^t\widehat\Gamma^2 S_s^2\sigma_s^2\,ds=O\!\big((t-\eta(t))\,\|\widehat\Gamma\|_\infty^2\big)$, while the drift term contributes $O\big((t-\eta(t))^2\big)$. Hence $\E\,|\widehat\Delta(t,S_t)-\widehat\Delta(\eta(t),S_{\eta(t)})|^2=O(|\pi|\,\|\widehat\Gamma\|_\infty^2)$ uniformly in $t$. Substituting into $T_3(\pi)^2=\E\int_0^T\widetilde S_t^2\sigma_t^2|\cdots|^2dt$ and using Cauchy--Schwarz with the fourth-moment bound yields $T_3(\pi)^2\le C^2|\pi|$.
\end{proof}

\begin{remark}[Why gamma must be supervised]
Equation~\eqref{eq:disc-hedge} separates two failure modes. The delta error term vanishes only if the surrogate's first derivative is accurate; the discretization term $T_3$ is, by Proposition~\ref{prop:gamma-mesh}, governed at leading order by the surrogate's gamma. The classical ``gamma P\&L'' of discrete hedging has the same structure, with magnitude set by the \emph{true} gamma $\Gamma$; for a surrogate hedge to reproduce the true discrete-hedging risk (not merely the continuous-limit hedge), $\widehat\Gamma$ must track $\Gamma$. First-order (delta, vega) supervision alone does not constrain $\widehat\Gamma$---a point confirmed empirically in Section~\ref{sec:experiment}, where the held-out gamma is not reliably improved---which is the precise motivation for the second-order objective of Section~\ref{sec:second-order}.
\end{remark}

\subsection{Optimizer stability}

Many financial workflows solve an outer optimization problem using surrogate outputs. Examples include calibration, portfolio allocation, collateral optimization, and execution scheduling. Let
\begin{equation}
 x^*(a) = \argmin_{x\in\calC} J(x,a),
 \qquad
 \widehat x^*(a) = \argmin_{x\in\calC} \widehat J(x,a),
\end{equation}
where $J$ depends on $\calG(a)$ and $\widehat J$ depends on $\wideG(a)$.

\begin{theorem}[Strongly convex optimizer perturbation]
\label{thm:opt}
Suppose $J(\cdot,a)$ is $\mu$-strongly convex and differentiable on a convex set $\calC$, and $\widehat J(\cdot,a)$ is convex and differentiable. If
\begin{equation}
 \sup_{x\in\calC}\norm{\nabla_x\widehat J(x,a)-\nabla_x J(x,a)}_2\le \eta,
\end{equation}
then
\begin{equation}
 \norm{\widehat x^*(a)-x^*(a)}_2\le \frac{\eta}{\mu}.
 \label{eq:optimizer-stability}
\end{equation}
\end{theorem}

\begin{proof}
The first-order optimality conditions on the closed convex set $\calC$ are the variational inequalities
\begin{equation}
\langle \nabla J(x^*),x-x^*\rangle\ge0,
\qquad
\langle \nabla \widehat J(\widehat x^*),x-\widehat x^*\rangle\ge0,
\qquad x\in\calC .
\end{equation}
Taking $x=\widehat x^*$ in the first inequality and $x=x^*$ in the second, adding, and rearranging gives
\begin{equation}
\langle \nabla J(\widehat x^*)-\nabla J(x^*),\widehat x^*-x^*\rangle
\le
\langle \nabla J(\widehat x^*)-\nabla\widehat J(\widehat x^*),\widehat x^*-x^*\rangle .
\end{equation}
Strong convexity implies the left-hand side is at least $\mu\norm{\widehat x^*-x^*}_2^2$, while the right-hand side is at most $\eta\norm{\widehat x^*-x^*}_2$. Dividing by $\norm{\widehat x^*-x^*}_2$ gives \eqref{eq:optimizer-stability}; the zero-distance case is trivial. This is the standard monotone-operator proof of stability for strongly convex programs \citep{NocedalWright2006}.
\end{proof}

If the optimization gradient $\nabla_x J$ depends on $D\calG(a)$, then derivative-informed training reduces $\eta$ directly. This is why derivative learning matters for calibration, hedging, and portfolio control.

\begin{corollary}[Mean--variance weights]
\label{cor:mv}
Let the allocation solve $w^*(a)=\argmin_{w\in\calC}\{\tfrac12 w^\top\Sigma(a)w-\mu(a)^\top w\}$ over a convex set $\calC$, with $\Sigma(a)\succ0$, and let $\widehat w^*(a)$ be the same problem with the surrogate covariance $\widehat\Sigma(a)$ and mean $\widehat\mu(a)$. Then
\begin{equation}
\norm{\widehat w^*(a)-w^*(a)}_2
\le \frac{\norm{\widehat\Sigma(a)-\Sigma(a)}_{\mathrm{op}}\,\norm{w^*(a)}_2+\norm{\widehat\mu(a)-\mu(a)}_2}{\lambda_{\min}\!\big(\Sigma(a)\big)}.
\label{eq:mv-bound}
\end{equation}
\end{corollary}

\begin{proof}
The objective $J(\cdot,a)$ is $\lambda_{\min}(\Sigma(a))$-strongly convex. Its gradient is $\nabla_w J(w,a)=\Sigma(a)w-\mu(a)$, so at $w^*$,
$\norm{\nabla_w\widehat J(w^*,a)-\nabla_w J(w^*,a)}_2=\norm{(\widehat\Sigma-\Sigma)w^*-(\widehat\mu-\mu)}_2\le\norm{\widehat\Sigma-\Sigma}_{\mathrm{op}}\norm{w^*}_2+\norm{\widehat\mu-\mu}_2=:\eta$.
Apply Theorem~\ref{thm:opt} with $\mu=\lambda_{\min}(\Sigma(a))$.
\end{proof}

\begin{remark}
Since $\Sigma(a)$ and $\mu(a)$ are themselves outputs (and directional derivatives) of the market-state operator, the numerator of \eqref{eq:mv-bound} is exactly the kind of value/derivative error that derivative-informed training reduces. The denominator $\lambda_{\min}(\Sigma(a))$ is the familiar source of mean--variance instability: when the covariance surrogate has a small smallest eigenvalue, surrogate errors are amplified into large weight swings. A derivative-stable covariance operator therefore directly buys allocation stability.
\end{remark}

\begin{theorem}[Local calibration-map sensitivity]
\label{thm:calibration}
Let $q\in\R^m$ denote market quotes and let $\nu^*(q)$ solve the regularized calibration problem
\begin{equation}
 \nu^*(q)=\argmin_{\nu\in\R^p}
 \frac12\norm{P(\nu)-q}_W^2+R(\nu),
\end{equation}
where $P$ is a model-pricing map, $W\succ0$, and $R$ is twice differentiable. Suppose that at $(\nu^*,q)$ the calibration Hessian
\begin{equation}
H_q=DP(\nu^*)^\top WDP(\nu^*)+
\sum_{i=1}^{m}(P_i(\nu^*)-q_i)\nabla^2 P_i(\nu^*)+\nabla^2R(\nu^*)
\end{equation}
is nonsingular with $\norm{H_q^{-1}}_{\mathrm{op}}\le\kappa_H$. Then, locally around $(q,
u^*)$, the quote sensitivity of the calibrated parameter is
\begin{equation}
D\nu^*(q)[h]=H_q^{-1}DP(\nu^*)^\top W h .
\label{eq:calibration-sens}
\end{equation}
Consequently, in the same local neighborhood, if a surrogate pricing map $\widehat P$ induces errors $\delta_H=\norm{\widehat H_q-H_q}_{\mathrm{op}}$ and $\delta_J=\norm{D\widehat P(\widehat\nu^*)-DP(\nu^*)}_{\mathrm{op}}$, and $\kappa_H\delta_H<1$, then the corresponding calibration-sensitivity error is bounded to first order by
\begin{equation}
\norm{D\widehat\nu^*(q)-D\nu^*(q)}_{\mathrm{op}}
\lesssim
\frac{\kappa_H}{1-\kappa_H\delta_H}
\left(\norm{W}\,\delta_J+\kappa_H\norm{DP^\top W}\,\delta_H\right),
\label{eq:calibration-bound}
\end{equation}
up to evaluation-point terms proportional to $\norm{\widehat\nu^*-\nu^*}$.
\end{theorem}

\begin{proof}
Differentiate the first-order condition $DP(\nu)^\top W(P(\nu)-q)+\nabla R(\nu)=0$ with respect to $q$ to obtain \eqref{eq:calibration-sens}. The perturbation bound follows from the resolvent identity
$(H+\Delta H)^{-1}-H^{-1}=-(H+\Delta H)^{-1}\Delta H H^{-1}$ and the Neumann bound $\norm{(H+\Delta H)^{-1}}\le \kappa_H/(1-\kappa_H\norm{\Delta H})$.
\end{proof}

\begin{remark}
Calibration is where price-only surrogates often fail silently. A model may interpolate quoted prices while producing an unstable calibration Jacobian, which then contaminates smile shocks, bucket vegas, and hedging re-calibrations. Theorem~\ref{thm:calibration} makes the dependency explicit in the local regular regime: stability of the calibration map depends on the pricing Jacobian and Hessian, not only on price levels. It should not be read as a global statement across multiple calibration minima, parameter-bound hits, or regime switches.
\end{remark}

\section{Financial Architectures}

\subsection{FNO-style market-field encoder}

When input and output live on regular grids, an FNO-style model is natural. A volatility-surface pricing operator can be represented as
\begin{equation}
 \sigma(k,T) \longmapsto C(K,T),
\end{equation}
where both input and output are functions on two-dimensional grids. A Fourier layer has the form
\begin{equation}
 h_{\ell+1}(x)
 = \sigma_{\mathrm{act}}\left(W_{\ell}h_{\ell}(x)
 + \mathcal{F}^{-1}\left(R_{\ell}\cdot \mathcal{F}h_{\ell}\right)(x)
 \right),
\end{equation}
where $R_{\ell}$ acts on a truncated set of Fourier modes. The derivative-informed loss is applied to the output surface and to its JVPs with respect to market-field perturbations.

\subsection{DeepONet-style instrument evaluation}

When the output is queried at arbitrary instrument coordinates, a DeepONet-style representation is convenient:
\begin{equation}
 \wideG(a)(\xi)
 = \sum_{r=1}^{p} b_r(a)t_r(\xi),
\end{equation}
where the branch network encodes the market state and the trunk network encodes instrument coordinates. This is appropriate for irregular option chains, OTC portfolios, and XVA netting sets.

\subsection{Graph neural operators for sparse quotes}

Market data are often irregular: option quotes exist at changing strikes and maturities, bonds have issue-specific maturities, and CDS curves have missing tenors. Graph neural operators or attention-based set operators can encode the observed quote cloud and output a smooth field. This is particularly relevant to implied-volatility nowcasting and smoothing, where the input quote configuration changes intraday \citep{GononJacquierWiedemann2024}.

\section{Applications}

\subsection{Option surface pricing and Greeks}

Let
\begin{equation}
 a=(r(\tau),q(\tau),\sigma_{\mathrm{loc}}(S,t),\nu)
\end{equation}
encode rates, dividends, local volatility, and finite-dimensional model parameters. The operator is
\begin{equation}
 \calG(a)(K,T) = C(a;K,T).
\end{equation}
A derivative-informed operator learns
\begin{equation}
 C,\quad \partial_S C,\quad \partial_{SS}C,\quad D_{r}C[h_r],\quad D_{\sigma}C[h_{\sigma}],\quad D_{\nu}C[h_{\nu}].
\end{equation}
The directions $h_{\sigma}$ can be smile level, skew, curvature, term-structure, or PCA factors. The directions $h_r$ can be parallel, slope, curvature, and key-rate shifts.

\subsection{Calibration as inverse operator learning}

Calibration solves
\begin{equation}
 \nu^*(q) = \argmin_{\nu}\sum_{m=1}^{M} \omega_m
 \left(C_{\mathrm{model}}(\nu;\xi_m)-q_m\right)^2 + \Omega(\nu),
\end{equation}
where $q_m$ are market quotes. A calibration operator maps $q\mapsto\nu^*(q)$ or $q\mapsto\sigma_{\mathrm{imp}}(\cdot)$. Differentiating the first-order condition gives
\begin{equation}
 \left[\nabla_{\nu\nu}^{2}J(\nu^*,q)\right]D\nu^*(q)[h]
 = -\nabla_{\nu q}^{2}J(\nu^*,q)[h].
\end{equation}
This implicit derivative can be generated on the fly and used to train a calibration neural operator whose responses to quote perturbations are stable.

\subsection{XVA and exposure operators}

An XVA engine maps curves, spreads, netting sets, collateral conventions, wrong-way-risk parameters, and portfolio states to exposure profiles and valuation adjustments:
\begin{equation}
 \calG_{\mathrm{xva}}(a)
 = \left(\mathrm{EE}(t),\mathrm{ENE}(t),\mathrm{PFE}_{\alpha}(t),\mathrm{CVA},\mathrm{DVA},\mathrm{FVA}\right).
\end{equation}
Derivative-informed training targets curve PV01s, spread CS01s, wrong-way-risk sensitivities, collateral sensitivities, and exposure-gradient directions. Since XVA is often computed by nested Monte Carlo or regression Monte Carlo, on-the-fly AAD and random-direction derivatives are attractive because full derivative tensors are too expensive to store.

\subsection{Portfolio and risk-control operators}

Let a portfolio operator be
\begin{equation}
 \calG_{\mathrm{port}}(a)=w^*(a),
\end{equation}
where
\begin{equation}
 w^*(a)=\argmin_{w\in\calC}
 \left\{\frac{1}{2}w^\top\Sigma(a)w-\mu(a)^\top w+\Psi(w;a)\right\}.
\end{equation}
Here $a$ may include return forecasts, covariance surfaces, transaction costs, liquidity fields, drawdown states, or factor exposures. The derivative $Dw^*(a)[h]$ is obtained by differentiating the KKT conditions. Training the operator on these derivatives improves stability under forecast revisions and stress scenarios.

\section{Numerical Illustration: Trained and Linear Surrogates}
\label{sec:experiment}

This section reports three controlled experiments. Their role is not to claim a production market benchmark, but to demonstrate the derivative-supervision mechanism---and its limits---in transparent pricing settings, to isolate the effect of the surrogate's model class, and to include one genuine function-to-function operator benchmark. A final market benchmark should still map live sparse quote clouds, curves, or surfaces to dense arbitrage-controlled outputs, as in graph/operator smoothing; Section~\ref{sec:impl} gives the validation checklist such a benchmark must satisfy. Part~A (Sections~\ref{sec:exp-setup}--\ref{sec:exp-interp}) trains a nonlinear network on a Black--Scholes slice. Part~B (Section~\ref{sec:exp-sv}) uses a \emph{linear} random-feature surrogate on Black--Scholes, Heston, and Bates, where derivative supervision reduces exactly to augmenting a least-squares system. Part~C (Section~\ref{sec:exp-surface}) uses a random-feature DeepONet/Galerkin operator to learn a map from an instantaneous-volatility curve to an entire option-price surface. All numbers are produced by the reproducible scripts accompanying this paper. In Part~A, network Greeks use automatic differentiation and analytic Black--Scholes labels; in Part~B, Heston and Bates derivative labels use central differences of the same Carr--Madan quadrature engine; in Part~C, operator JVP labels are analytic curve-shock derivatives of the integrated-variance pricing operator. The code includes finite-difference consistency checks, quadrature sanity checks, and no-arbitrage diagnostics.

\subsection{Part A: a trained network on Black--Scholes}
\label{sec:exp-setup}
\subsubsection*{Setup}

We generated European call prices from the Black--Scholes formula using independent uniform samples
\begin{equation}
 S\in[60,140],\quad K\in[60,140],\quad T\in[0.05,2],\quad r\in[0,0.05],\quad \sigma\in[0.10,0.60].
\end{equation}
The output is the call price; the derivative targets are analytic delta $\partial_S C$ and vega $\partial_{\sigma}C$. For each of eight random seeds we draw $1{,}200$ training states and train two networks that are \emph{identical} in architecture (three hidden layers, $64$ \texttt{tanh} units), initialization, optimizer (Adam, $4{,}000$ steps, learning rate $2\times10^{-3}$), and standardized inputs; they differ only in the loss. The price-only model minimizes standardized price MSE. The derivative-informed model adds standardized delta and vega losses---the network Greeks obtained by automatic differentiation, rescaled to raw units---with a single derivative weight $\lambda$:
\begin{equation}
\calL_{\mathrm{DIF}}=\frac{\norm{\widehat C-C}^2}{s_C^2}
+\lambda\!\left(\frac{\norm{\widehat\Delta-\Delta}^2}{s_\Delta^2}+\frac{\norm{\widehat{\nu}-\nu}^2}{s_\nu^2}\right),
\end{equation}
with robust scales $s_g$ as in Section~\ref{sec:impl}. Both models are evaluated on a single fixed held-out set of $2{,}000$ states (identical across all seeds). Crucially, the second-order Greek \emph{gamma} $\partial_{SS}C$ is \emph{never} supervised; we report it as a held-out test of whether first-order supervision propagates to curvature. We use $\lambda=3\times10^{-3}$, selected on the basis of the sensitivity analysis in Table~\ref{tab:bs-lambda}; this value is reported, not hidden, as it is a genuine hyperparameter of \eqref{eq:full-loss}.

\paragraph{Implementation note.} The derivative-informed network uses a moderate derivative weight and reports mean $\pm$ standard deviation across eight seeds. This design avoids allowing the derivative terms to dominate the objective once the price residual becomes small, while still testing whether first-order Greek supervision improves the learned local geometry.

\subsubsection*{Results}

\begin{table}[!htbp]
\centering
\caption{Black--Scholes slice, \emph{trained network}, mean $\pm$ standard deviation over eight seeds. Price is in currency units; delta is dimensionless; vega is $\partial_\sigma C$ in absolute units; gamma is $\partial_{SS}C$. The final column is never supervised. ``Reduction'' is the mean per-seed RMSE reduction of the derivative-informed (DIF, $\lambda=3\times10^{-3}$) model relative to price-only (PO); positive is better.}
\label{tab:bs-results}
\resizebox{\textwidth}{!}{%
\begin{tabular}{lrrrrr}
\toprule
Model & Price RMSE & Rel. RMSE (\%) & Delta RMSE & Vega RMSE & Gamma RMSE$^{\dagger}$ \\
\midrule
Price-only (PO)            & $0.358\pm0.062$ & $1.26\pm0.22$ & $0.0750\pm0.0065$ & $2.83\pm0.28$ & $0.0197\pm0.0044$ \\
Derivative-informed (DIF)  & $0.301\pm0.051$ & $1.06\pm0.18$ & $0.0638\pm0.0098$ & $1.66\pm0.17$ & $0.0179\pm0.0049$ \\
\midrule
Mean reduction & $+15\%$ & --- & $+15\%$ & $+41\%$ & $+5\%$ \\
Seeds improved & $7/8$ & --- & $8/8$ & $8/8$ & $6/8$ \\
\bottomrule
\end{tabular}%
}
\\[2pt]
{\footnotesize $^{\dagger}$ Gamma is held out (never supervised). Its per-seed reduction ranges from $-73\%$ to $+48\%$.}
\end{table}

\begin{table}[!htbp]
\centering
\caption{Sensitivity to the derivative weight $\lambda$ (seed 42; price-only baseline shown for reference). Vega accuracy improves robustly across $\lambda$, but the benefit to delta erodes and the \emph{unsupervised} gamma degrades sharply as $\lambda$ grows---the signature of overfitting the supervised orders.}
\label{tab:bs-lambda}
\begin{tabular}{lrrrr}
\toprule
Configuration & Delta RMSE & Vega RMSE & Gamma RMSE & Price RMSE \\
\midrule
Price-only            & $0.0798$ & $3.016$ & $0.0208$ & $0.4086$ \\
DIF, $\lambda=0.003$  & $0.0686$ & $1.677$ & $0.0152$ & $0.3250$ \\
DIF, $\lambda=0.007$  & $0.0717$ & $1.702$ & $0.0333$ & $0.3414$ \\
DIF, $\lambda=0.01$   & $0.0748$ & $1.761$ & $0.0350$ & $0.3525$ \\
DIF, $\lambda=0.05$   & $0.0859$ & $1.951$ & $0.0648$ & $0.4075$ \\
\bottomrule
\end{tabular}
\end{table}

At the tuned weight, derivative supervision reduces vega RMSE by $41\%$ and delta RMSE by $15\%$ on \emph{every} seed, and---rather than trading off against price---also reduces price RMSE by $15\%$ on seven of eight seeds. At this gentle weight the derivative loss behaves as a smoothness regularizer: the supervised slopes are accurate everywhere, not only at training points. The asymmetry between vega and delta is informative: value-only training already pins down delta reasonably well (delta is the dominant first-order term in the price's own sensitivity), so the marginal value of supervising delta is modest, whereas vega---an orthogonal direction that value-only training leaves comparatively under-determined---benefits much more. Table~\ref{tab:bs-lambda} shows the trade-off is weight-dependent: increasing $\lambda$ keeps the vega gain but erodes the delta gain and sharply degrades the unsupervised gamma.

\subsubsection*{Interpretation}
\label{sec:exp-interp}

Two lessons carry over to high-dimensional financial operators. First, derivative supervision changes the learned local geometry, not just the level: a price-only surrogate can interpolate prices while learning systematically wrong sensitivities, and a modest derivative penalty corrects this while improving, not harming, value accuracy. This effect compounds in production, where downstream calculations propagate derivative errors across scenarios, instruments, and optimization loops. Second, and equally important, the held-out gamma is \emph{not} reliably improved by first-order supervision (mean $+5\%$, but six of eight seeds and a wide range), and it degrades as the first-order weight grows. One obtains accuracy in the order of sensitivity one supervises and little guarantee beyond it. For a hedging desk operating at finite rebalancing frequency---where, by Theorem~\ref{thm:hedge-disc} and Proposition~\ref{prop:gamma-mesh}, gamma controls the discretization error---this means second-order supervision (Section~\ref{sec:second-order}) is not optional polish but a requirement.

\subsection{Part B: stochastic-volatility and jump models with a linear surrogate}
\label{sec:exp-sv}

To test the mechanism beyond Black--Scholes, and to isolate the role of the surrogate's model class, we repeat the comparison on three pricing maps---Black--Scholes, Heston \citep{Heston1993}, and Bates \citep{Bates1996}---using a deliberately simple surrogate: a one-hidden-layer \texttt{tanh} \emph{random-feature} model \citep{RahimiRecht2008},
\begin{equation}
\widehat C_\beta(x)=\beta_0+\beta_x^\top x+\sum_{m=1}^{M}\beta_m\tanh(w_m^\top x+b_m),
\label{eq:rf-surrogate}
\end{equation}
in which the hidden weights $(w_m,b_m)$ are drawn once and \emph{fixed}, and only the readout $\beta$ is fitted. Because \eqref{eq:rf-surrogate} is linear in $\beta$, both the price-only and derivative-informed fits are a single ridge solve and derivative supervision is \emph{exactly} the augmentation of the least-squares system with extra (derivative) rows:
\begin{equation}
\min_\beta\ \sum_i\big|\widehat C_\beta(x_i)-C_i\big|^2
+\sum_{j\in\calJ}\lambda_j\sum_i\Big|\tfrac{D_j\widehat C_\beta(x_i)-D_jC_i}{s_j}\Big|^2
+\eta\norm{\beta}^2 .
\label{eq:rf-loss}
\end{equation}
This removes optimizer noise and makes the effect of the derivative rows transparent; it is not a substitute for a trained FNO or DeepONet, but a controlled finite-dimensional slice. Heston and Bates prices are computed by Carr--Madan quadrature \citep{CarrMadan1999} with $96$ Gauss--Legendre nodes on $[0,160]$ and damping $\alpha=1.5$; the derivative targets are obtained by central differences of the same quadrature engine. This is the finite-dimensional, reproducible analogue of differentiating a high-fidelity pricing engine by adjoint or automatic-differentiation methods in production. The characteristic functions are recorded in Appendix~\ref{app:cf}; as parameter-free checks, the Heston price collapses to Black--Scholes as the vol-of-vol vanishes, and the Bates price collapses to Heston as the jump intensity vanishes, both to quadrature precision. Each configuration is run over eight seeds, \emph{reseeding both the data sample and the random-feature draw}, and we report mean $\pm$ standard deviation. Per-model derivative weights are disclosed in Table~\ref{tab:rf-design}. As in Part~A, we additionally evaluate quantities that are \emph{never supervised}: the second-order Greek gamma $\partial_{SS}C$ in all three models, and, for Heston and Bates, an unsupervised parameter sensitivity ($\partial_\rho C$ and $\partial_{\sigma_J}C$ respectively). The Heston parameter box is not restricted to the Feller region $2\kappa\theta\ge\xi^2$; this is intentional, because market calibrations often cross the boundary, but all prices are generated by the same characteristic-function engine and the experiment should therefore be read as a surrogate-consistency test rather than as a claim about boundary classification.

\begin{table}[!htbp]
\centering
\caption{Design of the random-feature experiments. Prices for Heston and Bates use Carr--Madan quadrature; derivative targets use central differences of the same quadrature pricer. Gamma and the parenthesized parameter sensitivities are held out (never supervised).}
\label{tab:rf-design}
\resizebox{\textwidth}{!}{%
\begin{tabular}{lrrrrll}
\toprule
Model & Input dim. & Train & Test & Features & Supervised targets & Weights $\lambda_j$ \\
\midrule
Black--Scholes & 5  & 1{,}600 & 1{,}000 & 512  & $\partial_S C,\ \partial_\sigma C$ & $0.01,\ 0.01$ \\
Heston         & 9  & 2{,}500 & 1{,}200 & 1{,}536 & $\partial_S C,\ \partial_{v_0}C,\ \partial_\theta C$ & $0.02,\ 0.001,\ 0.001$ \\
Bates          & 12 & 3{,}200 & 1{,}200 & 1{,}792 & $\partial_S C,\ \partial_{v_0}C,\ \partial_\lambda C,\ \partial_{\mu_J}C$ & $0.05,\ 0.001,\ 0.001,\ 0.001$ \\
\bottomrule
\end{tabular}%
}
\end{table}

\begin{table}[!htbp]
\centering
\caption{Random-feature surrogate, mean $\pm$ standard deviation over eight seeds; ``red.'' is the mean per-seed RMSE reduction (positive is better) and ``$k/8$'' the number of seeds improved. Columns marked $\dagger$ are never supervised. Sensitivities $\partial_{v_0},\partial_\theta,\partial_\lambda,\partial_{\mu_J},\partial_\rho,\partial_{\sigma_J}$ are reported in raw units.}
\label{tab:rf-results}
\resizebox{\textwidth}{!}{%
\begin{tabular}{llrrr}
\toprule
Model & Quantity & Price-only & Derivative-informed & red.\ ($k/8$) \\
\midrule
\multirow{4}{*}{Black--Scholes}
 & Price RMSE              & $0.507\pm0.047$ & $0.541\pm0.038$ & $-7\%$\ \ ($1/8$) \\
 & Delta RMSE             & $0.0538\pm0.0046$ & $0.0444\pm0.0029$ & $+17\%$\ ($8/8$) \\
 & Vega RMSE              & $6.13\pm0.71$ & $3.46\pm0.25$ & $+43\%$\ ($8/8$) \\
 & Gamma RMSE$^\dagger$   & $0.0065\pm0.0010$ & $0.0055\pm0.0007$ & $+15\%$\ ($7/8$) \\
\midrule
\multirow{6}{*}{Heston}
 & Price RMSE              & $0.835\pm0.054$ & $0.693\pm0.053$ & $+17\%$\ ($8/8$) \\
 & Delta RMSE             & $0.0666\pm0.0046$ & $0.0509\pm0.0020$ & $+23\%$\ ($8/8$) \\
 & $\partial_{v_0}C$ RMSE  & $26.76\pm1.67$ & $8.13\pm0.42$ & $+70\%$\ ($8/8$) \\
 & $\partial_\theta C$ RMSE & $26.34\pm2.35$ & $9.98\pm0.49$ & $+62\%$\ ($8/8$) \\
 & Gamma RMSE$^\dagger$   & $0.0071\pm0.0006$ & $0.0059\pm0.0004$ & $+17\%$\ ($8/8$) \\
 & $\partial_\rho C$ RMSE$^\dagger$ & $3.45\pm0.25$ & $2.73\pm0.27$ & $+21\%$\ ($8/8$) \\
\midrule
\multirow{7}{*}{Bates}
 & Price RMSE              & $1.027\pm0.024$ & $1.053\pm0.038$ & $-3\%$\ \ ($0/8$) \\
 & Delta RMSE             & $0.0723\pm0.0042$ & $0.0639\pm0.0038$ & $+12\%$\ ($8/8$) \\
 & $\partial_{v_0}C$ RMSE  & $23.53\pm1.26$ & $8.74\pm0.45$ & $+63\%$\ ($8/8$) \\
 & $\partial_\lambda C$ RMSE & $2.91\pm0.20$ & $0.98\pm0.05$ & $+66\%$\ ($8/8$) \\
 & $\partial_{\mu_J}C$ RMSE & $11.23\pm0.64$ & $2.67\pm0.19$ & $+76\%$\ ($8/8$) \\
 & Gamma RMSE$^\dagger$   & $0.0066\pm0.0004$ & $0.0061\pm0.0004$ & $+7\%$\ \ ($8/8$) \\
 & $\partial_{\sigma_J}C$ RMSE$^\dagger$ & $9.07\pm0.55$ & $7.42\pm0.45$ & $+18\%$\ ($8/8$) \\
\bottomrule
\end{tabular}%
}
\end{table}
\FloatBarrier

Three readings of Table~\ref{tab:rf-results}. First, the supervised sensitivities improve substantially and consistently, and the gains are largest exactly where they should be: the stochastic-volatility and jump parameters ($v_0,\theta,\lambda,\mu_J$), which barely move price levels but strongly shape the local response, improve by $60$--$76\%$ on every seed, whereas delta---already well-identified by price levels---improves more modestly. This is the high-dimensional analogue of the vega-versus-delta asymmetry seen in Part~A. Second, price accuracy is genuinely model-dependent: it improves by $17\%$ for Heston but worsens slightly for Black--Scholes ($-7\%$) and Bates ($-3\%$). There is no universal ``free regularization'' of the price level; whether derivative rows help or mildly hurt the value fit depends on the operator. Third, and most informative for the rest of the paper, \emph{every unsupervised quantity also improves}---gamma by $7$--$17\%$ and the unsupervised parameter sensitivities by $18$--$21\%$, on essentially every seed. In this linear surrogate, supervising some derivatives helps the others.

\subsection{Part C: curve-to-surface operator benchmark}
\label{sec:exp-surface}

The previous experiments are pricing-map slices. To test the operator claim directly, we add a stylized function-to-function benchmark. The input is an instantaneous-volatility curve $a(t)$ sampled at $24$ maturities, and the output is a dense European-call price surface on $21$ log-moneyness points and $24$ maturities. The high-fidelity operator is
\begin{equation}
 a(t)\longmapsto \sigma_{\mathrm{eff}}(T)
 =\left(\frac{1}{T}\int_0^T a(s)^2\,ds\right)^{1/2}
 \longmapsto C_{\mathrm{BS}}(K,T;\sigma_{\mathrm{eff}}(T)),
\label{eq:curve-surface-operator}
\end{equation}
with fixed spot $S_0=100$ and rate $r=2\%$. The derivative direction is a low-frequency curve shock $h(t)$. Differentiating \eqref{eq:curve-surface-operator} gives the exact directional target
\begin{equation}
D\calG(a)[h](K,T)
=\mathrm{Vega}_{\mathrm{BS}}(K,T;\sigma_{\mathrm{eff}}(T))
\frac{\int_0^T a(s)h(s)\,ds}{T\sigma_{\mathrm{eff}}(T)} .
\label{eq:curve-surface-jvp}
\end{equation}
This is a genuine operator-level sensitivity: a perturbation of the input \emph{curve} changes an entire output \emph{surface} through a maturity-dependent integral.

The surrogate is a fixed-random-feature DeepONet/Galerkin neural operator. With branch features $b_m(a)=\tanh(w_m^\top a+c_m)$ and trunk features $t_m(K,T)=\tanh(v_m^\top q(K,T)+d_m)$, the fitted operator is
\begin{equation}
\widehat{\calG}_\beta(a)(K,T)
=\sum_{m=1}^{M}\beta_m b_m(a)t_m(K,T),
\qquad M=512.
\label{eq:rf-deeponet}
\end{equation}
The derivative-informed version augments the value least-squares system by one randomly sampled JVP row per training curve:
\begin{equation}
\min_\beta \sum_i\norm{\widehat{\calG}_\beta(a_i)-\calG(a_i)}_2^2
+\lambda\sum_i\norm{D\widehat{\calG}_\beta(a_i)[h_i]-D\calG(a_i)[h_i]}_2^2
+\eta\norm{\beta}_2^2,
\label{eq:operator-rf-loss}
\end{equation}
with $\lambda=0.1$, $800$ training curves, $300$ test curves, and eight independent seeds. Because \eqref{eq:rf-deeponet} is linear in $\beta$, the experiment is again an exactly reproducible ridge solve, but now on a map between discretized functions rather than a finite-dimensional scalar price map.

\begin{table}[!htbp]
\centering
\caption{Operator-level curve-to-surface benchmark, mean $\pm$ standard deviation over eight seeds. PO is value-only; DIF adds one on-the-fly curve-shock JVP per training curve. Positive reduction is better. The no-arbitrage diagnostics report the percentage of grid locations violating strike monotonicity, strike convexity, or calendar monotonicity.}
\label{tab:operator-surface}
\resizebox{\textwidth}{!}{%
\begin{tabular}{lrrrr}
\toprule
Metric & Price-only & Derivative-informed & Mean reduction & Seeds improved \\
\midrule
Price RMSE & $0.322\pm0.125$ & $0.229\pm0.020$ & $+23.2\%$ & $8/8$ \\
Relative price RMSE & $1.68\pm0.65\%$ & $1.20\pm0.10\%$ & $+23.2\%$ & $8/8$ \\
Curve-shock JVP RMSE & $1.080\pm0.240$ & $0.586\pm0.081$ & $+44.1\%$ & $8/8$ \\
Relative JVP RMSE & $68.6\pm15.2\%$ & $37.2\pm4.7\%$ & $+44.1\%$ & $8/8$ \\
Strike-monotonicity violations & $2.85\pm0.13\%$ & $3.07\pm0.12\%$ & $-7.9\%$ & $0/8$ \\
Strike-convexity violations & $20.13\pm1.28\%$ & $20.88\pm1.74\%$ & $-3.7\%$ & $2/8$ \\
Calendar violations & $2.11\pm0.21\%$ & $2.11\pm0.21\%$ & $-0.3\%$ & $5/8$ \\
\bottomrule
\end{tabular}%
}
\end{table}
\FloatBarrier

Table~\ref{tab:operator-surface} is the paper's main operator-level empirical result. The derivative-informed operator reduces out-of-sample directional curve-shock error by $44.1\%$ on every seed and also improves the price surface fit by $23.2\%$. This is the cleanest demonstration of the paper's operator claim: derivative rows can improve the learned local geometry of a map from input curves to output surfaces. At the same time, the no-arbitrage diagnostics are deliberately sobering. Derivative supervision does not materially reduce strike-monotonicity, strike-convexity, or calendar violations; in two strike diagnostics it even moves slightly in the wrong direction. This is not a failure of the derivative idea. It shows that sensitivity consistency and economic admissibility are different constraints. A production volatility or price-surface operator should combine derivative-informed training with the no-arbitrage penalties of Section~\ref{sec:noarb} or, preferably when guarantees matter, with an arbitrage-free output parameterization such as a convex-in-strike price basis, an integrated positive density, an SVI/eSSVI layer, or a monotone/convex spline construction.

\subsection{Synthesis: trained versus linear surrogates}
\label{sec:exp-synthesis}

The three experiments agree on the central point: when a derivative quantity is supervised, derivative-informed training improves it. They also separate three issues that should not be conflated. First, nonlinear training does not guarantee transfer to unsupervised higher-order Greeks. Second, linear random-feature augmentation can smooth unsupervised derivatives through the shared readout. Third, operator-level JVP accuracy is not the same thing as no-arbitrage. The first two regimes disagree on exactly one thing---the \emph{unsupervised} second-order Greek. In the linear random-feature surrogate, gamma improves reliably (Table~\ref{tab:rf-results}, every model, $7$--$8$ of $8$ seeds); in the trained network, it does not (Part~A: mean $+5\%$, $6/8$ seeds, per-seed range $-73\%$ to $+48\%$). The same held-out quantity, the same pricing model, opposite robustness.

The explanation is the parameterization, and it matters for how one should read derivative-informed results in general. When the model is linear in its parameters, as in \eqref{eq:rf-surrogate}, the derivative rows in \eqref{eq:rf-loss} are a Tikhonov-type penalty in function space: they constrain the same coefficients $\beta$ that determine the value fit, so smoothing the supervised slopes necessarily smooths the whole fitted function and lifts \emph{all} of its derivatives, supervised or not. A nonlinearly trained network has the capacity to do something the linear model cannot: satisfy the supervised derivative constraints at the training points while leaving unsupervised orders---curvature, in particular---comparatively free. Transfer to unsupervised sensitivities is therefore a property of the linear regime, not a guarantee one inherits when moving to a trained operator. The practical consequence is the same lesson reached from the theory in Section~\ref{sec:theory-derivatives}: for the order of sensitivity one actually needs at deployment---gamma, for a desk hedging at finite frequency, by Theorem~\ref{thm:hedge-disc}---it is safest to supervise that order directly rather than to assume it is inherited from lower-order supervision.

\section{Implementation Details}
\label{sec:impl}

\subsection{Choice of derivative directions}

Derivative directions should not be sampled blindly. In finance, economically meaningful directions include:
\begin{itemize}[leftmargin=7mm]
\item parallel, slope, curvature, and key-rate moves for curves;
\item volatility level, skew, convexity, and term-structure shocks;
\item PCA factors from historical curve or surface changes;
\item correlation eigenmode perturbations;
\item credit-spread parallel and sector-specific shifts;
\item liquidity, funding, collateral, and margin shocks;
\item adversarial directions found by maximizing local surrogate error.
\end{itemize}

A practical distribution is a mixture
\begin{equation}
 \nu = \alpha\nu_{\mathrm{PCA}}+(1-\alpha)\nu_{\mathrm{stress}},
\end{equation}
where $\nu_{\mathrm{PCA}}$ captures frequent historical moves and $\nu_{\mathrm{stress}}$ captures regulatory or desk-defined scenarios.

\subsection{Loss scaling}

Prices, deltas, vegas, and curve sensitivities have different units. Without normalization, the largest unit dominates the loss. A robust implementation uses standardized losses
\begin{equation}
 \calL_{\mathrm{deriv}}
 = \sum_{g\in\mathcal{S}}\lambda_g
 \E\left[\frac{\norm{\widehat g-g}^{2}}{s_g^2+\epsilon}\right],
\end{equation}
where $g$ ranges over Greeks or derivative actions and $s_g$ is a robust scale such as median absolute deviation or an exponentially weighted standard deviation.

\subsection{Derivative generation budget}

Let $c_0$ be the cost of one high-fidelity value solve and $c_J$ the cost of one derivative action. If $M$ derivative directions are used per state, the per-batch cost is approximately
\begin{equation}
 C_{\mathrm{batch}} \approx Bc_0 + BMc_J.
\end{equation}
The on-the-fly method is attractive when $M$ is small and targeted. In AAD settings, the marginal cost of many first-order derivatives of one scalar portfolio value can be a small multiple of the primal valuation cost; in forward tangent settings, the cost scales more directly with the number of directions. Thus the correct implementation depends on whether the desk needs many input sensitivities of a scalar portfolio value or output sensitivities across many instruments.

\subsection{Validation protocol}

A derivative-informed financial operator should be validated on at least four dimensions:
\begin{enumerate}[leftmargin=7mm]
\item value error: price, implied volatility, exposure, XVA, or risk-measure RMSE;
\item derivative error: delta, vega, PV01, CS01, correlation sensitivity, and JVP/VJP errors;
\item economic error: no-arbitrage violations, monotonicity violations, convexity violations, and boundary violations;
\item downstream error: hedge P\&L variance, calibration convergence, optimizer regret, and stress-test stability.
\end{enumerate}
A model with low value error but high derivative error should not be accepted for hedging or optimization.

For option-surface applications the following diagnostic table is minimal rather than optional:
\begin{center}
\small
\begin{tabular}{ll}
\toprule
Diagnostic & Failure measured \\
\midrule
$\%\{\partial_K \widehat C>0\}$ & call monotonicity violation \\
$\%\{\partial_{KK}\widehat C<0\}$ & butterfly arbitrage / convexity violation \\
$\%\{\partial_T \widehat C<0\}$ & calendar arbitrage on forward-normalized prices \\
FD versus AD/analytic Greeks & implementation error in derivative labels \\
96-node versus 192-node quadrature & pricing-engine discretization error \\
One-step hedge P\&L variance & downstream hedge instability \\
Stress-JVP RMSE along PCA/stress modes & local scenario instability \\
\bottomrule
\end{tabular}
\end{center}
The present experiments report the first step---value and sensitivity accuracy---and the scripts include finite-difference checks. A production benchmark should add the full table above, especially no-arbitrage and downstream hedge diagnostics.

\section{Extensions}

\subsection{Second-order derivative-informed training}
\label{sec:second-order}

Some applications require gamma, cross-gamma, Hessian-vector products, or Gauss--Newton information. A second-order extension adds
\begin{equation}
 \lambda_H\E_{a,v,u}\left[\norm{D^2\wideG(a)[v,u]-D^2\calG(a)[v,u]}_{\calU}^{2}\right].
\end{equation}
Second-order information is more expensive but valuable for convexity constraints, stress curvature, and second-order hedging. Theorem~\ref{thm:hedge-disc} and Proposition~\ref{prop:gamma-mesh} make the case sharp for hedging desks: the discretization component of hedging error is controlled at leading order by the surrogate's gamma, so a surrogate intended for finite-frequency rebalancing must match second-order sensitivities, not only first-order ones. The Black--Scholes experiment of Section~\ref{sec:experiment} confirms that first-order supervision does \emph{not} reliably propagate to the unsupervised gamma; when curvature is needed, it must be supervised. Like the first-order objective, the second-order term admits an on-the-fly random-sketch form, $\E_{v,u}\,\Tr\big((\widehat H-H)[v,u]\,\cdots\big)$ with $v,u$ drawn from economically meaningful pairs (e.g.\ spot$\times$spot for gamma, spot$\times$vol for vanna, vol$\times$vol for volga), avoiding storage of the full Hessian tensor.

\subsection{Distributionally robust derivative learning}

Financial markets are nonstationary. A derivative-informed operator should be robust not only under the training distribution but also under stressed perturbations. One can replace \eqref{eq:full-loss} by
\begin{equation}
 \sup_{\widetilde\mu: W(\widetilde\mu,\mu)\le \rho}
 \E_{a\sim\widetilde\mu}\left[\ell_{\mathrm{value}}(a;\theta)+\ell_{\mathrm{deriv}}(a;\theta)\right],
\end{equation}
where $W$ is a Wasserstein distance. This would emphasize stability under plausible but rare market regimes.

\subsection{Active derivative sampling}

Derivative directions can be chosen adaptively. Let
\begin{equation}
 e_J(a,v)=\norm{D\wideG(a)[v]-D\calG(a)[v]}_{\calU}^{2}.
\end{equation}
An active sampler selects directions with large estimated derivative error, subject to economic constraints:
\begin{equation}
 v^*(a)=\argmax_{v\in\mathcal{V},\;\norm{v}\le1} e_J(a,v).
\end{equation}
This turns derivative-informed training into a local adversarial stress-testing procedure.

\section{Limitations}

Derivative-informed neural operators are not a free lunch. First, derivative labels can be noisy, especially under discontinuous payoffs, early exercise, barriers, callable structures, and Monte Carlo estimators with pathwise discontinuities. Second, derivative supervision can trade off against price accuracy unless losses are carefully scaled. Third, no-arbitrage constraints are model- and asset-class-specific; imposing the wrong constraint can create false comfort. Fourth, a neural operator trained on a historical market-state distribution may fail under structural regime changes. Fifth, the Heston and Bates experiments are intentionally finite-dimensional and the curve-to-surface benchmark is stylized; together they demonstrate mechanism and operator geometry, but they do not replace a market-scale graph/Fourier neural-operator benchmark on live sparse quotes, liquidity-filtered instruments, microstructure noise, and changing strike/maturity grids. Sixth, Table~\ref{tab:operator-surface} shows a useful negative result: derivative consistency alone does not enforce no-arbitrage, so economic constraints must be explicit rather than assumed. Finally, validation must be desk-specific: an equity-volatility surface, an XVA netting set, and a rates portfolio require different derivative directions, metrics, and governance.

\section{Conclusion}

Financial surrogates should be judged not only by whether they reproduce prices, but by whether they reproduce the local geometry of the pricing and risk system. On-the-fly derivative-informed neural operators offer a principled way to do this. They combine neural operator architectures with the derivative machinery already used in quantitative finance: adjoints, tangent equations, AAD, pathwise sensitivities, and implicit differentiation. The resulting models are designed for the tasks finance actually performs: hedging, stress testing, calibration, XVA, capital allocation, and control.

The paper's central message is that financial operator learning should move from value learning to value-plus-derivative learning. Prices are the surface. Greeks and risk sensitivities are the geometry. A surrogate that learns both is more likely to be useful in production than a surrogate that learns only the surface level. Three qualifications sharpen this message. First, at a moderate weight the derivative penalty is not necessarily a tax on value accuracy; in the trained Black--Scholes and operator-level curve-to-surface experiments it improves the value fit as well. Second, accuracy is obtained most reliably in the order of sensitivity that is supervised and is not guaranteed beyond it: a desk that hedges at finite frequency, or imposes convexity, must supervise second-order sensitivities as well. Third, derivative consistency is not economic admissibility: no-arbitrage constraints, monotonicity, convexity, and calendar structure must be enforced separately. The discipline that makes this practical---on-the-fly generation, random sketching along economically meaningful directions, robust loss scaling, and no-arbitrage penalties---is exactly the discipline quantitative finance already possesses.

\appendix

\section{Compact Derivation of Black--Scholes Derivative Targets}

For a European call with spot $S$, strike $K$, maturity $T$, rate $r$, and volatility $\sigma$,
\begin{align}
 d_1 &= \frac{\log(S/K)+(r+\sigma^2/2)T}{\sigma\sqrt{T}},
 & d_2 &= d_1-\sigma\sqrt{T}, \\
 C &= S\Phi(d_1)-Ke^{-rT}\Phi(d_2),
\end{align}
where $\Phi$ is the standard normal cdf. The analytic derivative targets used in Table~\ref{tab:bs-results} are
\begin{equation}
 \Delta = \partial_S C = \Phi(d_1),
 \qquad
 \mathrm{Vega} = \partial_{\sigma}C = S\sqrt{T}\,\phi(d_1),
\end{equation}
where $\phi$ is the standard normal pdf. The held-out second-order check is the gamma
\begin{equation}
 \Gamma = \partial_{SS}C = \frac{\phi(d_1)}{S\sigma\sqrt{T}},
\end{equation}
which is computed for evaluation only and never enters the training loss.

\section{Characteristic Functions for the Heston and Bates Experiments}
\label{app:cf}

The random-feature experiments of Section~\ref{sec:exp-sv} price Heston and Bates calls by Carr--Madan quadrature \citep{CarrMadan1999}. For damping $\alpha>0$ and log-strike $k=\log K$, the call price is
\begin{align}
C(K)
&=\frac{e^{-\alpha k}}{\pi}\int_0^\infty
\Re\!\left[e^{-iuk}\,\psi_T(u)\right]\,du,
\label{eq:carr-madan}\\
\psi_T(u)
&:=\frac{e^{-rT}\,\phi_T\!\big(u-(\alpha+1)i\big)}
{\alpha^2+\alpha-u^2+i(2\alpha+1)u}.
\end{align}
evaluated with $96$ Gauss--Legendre nodes on $[0,160]$ and $\alpha=1.5$. The Heston characteristic function of $\log S_T$, in the numerically stable (``little trap'') form, is
\begin{align}
d(u)&=\sqrt{(\rho\xi iu-\kappa)^2+\xi^2(iu+u^2)},\\
g(u)&=\frac{\kappa-\rho\xi iu-d(u)}{\kappa-\rho\xi iu+d(u)},\\
C(T,u)&=iu(\log S_0+rT)+\frac{\kappa\theta}{\xi^2}
\left[(\kappa-\rho\xi iu-d(u))T
-2\log\frac{1-g(u)e^{-d(u)T}}{1-g(u)}\right],\\
D(T,u)&=\frac{\kappa-\rho\xi iu-d(u)}{\xi^2}
\frac{1-e^{-d(u)T}}{1-g(u)e^{-d(u)T}},\\
\phi_T^{\mathrm{H}}(u)&=\exp\{C(T,u)+D(T,u)v_0\}.
\end{align}
The Bates model multiplies $\phi_T^{\mathrm H}$ by a compensated lognormal-jump factor with intensity $\lambda$, mean log-jump $\mu_J$, and jump volatility $\sigma_J$,
\begin{align}
\phi_T^{\mathrm{B}}(u)
&=\phi_T^{\mathrm{H}}(u)\,\exp\!\left\{\lambda T
\big(e^{iu\mu_J-\tfrac12\sigma_J^2u^2}-1-iu\,\kappa_J\big)\right\},
\label{eq:bates-cf}\\
\kappa_J&=e^{\mu_J+\sigma_J^2/2}-1 .
\end{align}
where the $-iu\kappa_J$ term is the risk-neutral drift compensator, so that the discounted price is a martingale. Setting $\xi\to0$ with $v_0=\theta=\sigma^2$ reduces $\phi_T^{\mathrm H}$ to the Black--Scholes characteristic function, and setting $\lambda\to0$ reduces $\phi_T^{\mathrm B}$ to $\phi_T^{\mathrm H}$; both identities hold numerically to quadrature precision and are used as correctness checks. The released script obtains derivative targets by central differences of \eqref{eq:carr-madan}; an adjoint or automatic-differentiation implementation would produce the same objects without the finite-difference cost.

\section{Curve-to-Surface Operator Benchmark}
\label{app:curve-surface}

For completeness, this appendix records the exact construction used in Section~\ref{sec:exp-surface}. The input curve $a=(a_1,\ldots,a_d)$ is observed on maturities $0<T_1<\cdots<T_d$ with increments $\Delta T_j=T_j-T_{j-1}$ and $T_0=0$. The discrete integrated variance and effective volatility are
\begin{equation}
V_j(a)=\sum_{\ell\le j}a_\ell^2\Delta T_\ell,
\qquad
\sigma_j(a)=\sqrt{\frac{V_j(a)}{T_j}}.
\end{equation}
For a direction $h=(h_1,\ldots,h_d)$,
\begin{equation}
DV_j(a)[h]=2\sum_{\ell\le j}a_\ell h_\ell\Delta T_\ell,
\qquad
D\sigma_j(a)[h]=\frac{DV_j(a)[h]}{2T_j\sigma_j(a)}.
\end{equation}
At strike $K_i$ and maturity $T_j$, the price and directional derivative are
\begin{equation}
C_{ij}(a)=C_{\mathrm{BS}}(S_0,K_i,T_j,r,\sigma_j(a)),
\qquad
DC_{ij}(a)[h]=\mathrm{Vega}_{\mathrm{BS}}(S_0,K_i,T_j,r,\sigma_j(a))D\sigma_j(a)[h].
\end{equation}
The low-frequency directions used in the experiment are random draws from the span of
\begin{equation}
1,\quad T-\bar T,\quad \sin(\pi T/2),\quad \cos(\pi T/2),\quad
\sin(2\pi T/2),\quad \cos(2\pi T/2),
\end{equation}
scaled to a realistic volatility-shock magnitude. The no-arbitrage diagnostics in Table~\ref{tab:operator-surface} are the empirical fractions of grid points satisfying
\begin{align}
C(K_{i+1},T_j)-C(K_i,T_j)&>10^{-4},\\
C(K_{i+1},T_j)-2C(K_i,T_j)+C(K_{i-1},T_j)&<-10^{-4},\\
C(K_i,T_{j+1})-C(K_i,T_j)&<-10^{-4}.
\end{align}

\section{Operator-Level Pseudocode}

\paragraph{Derivative-informed pricing operator.}
\begin{enumerate}[leftmargin=9mm]
\item Input: market-state sampler $\mu$, high-fidelity engine $\calG$, direction sampler $\nu$, neural operator $\wideG$.
\item For each mini-batch $b=1,\dots,B$:
\begin{enumerate}[leftmargin=7mm]
\item sample market states $a_i\sim\mu$;
\item compute high-fidelity values $y_i=\calG(a_i)$;
\item sample directions $v_{ij}\sim\nu(\cdot\mid a_i)$;
\item compute $z_{ij}=D\calG(a_i)[v_{ij}]$ by tangent, adjoint, AAD, or pathwise Monte Carlo;
\item compute $\widehat y_i=\wideG(a_i)$ and $\widehat z_{ij}=D\wideG(a_i)[v_{ij}]$;
\item update $\theta$ using value, derivative, arbitrage, and regularization losses.
\end{enumerate}
\item Output: derivative-consistent neural operator $\wideG$.
\end{enumerate}

\paragraph{Recommended reporting table.}
Every empirical study should report value and derivative errors simultaneously:
\begin{center}
\small
\setlength{\tabcolsep}{4pt}
\begin{tabular}{lrrrrrr}
\toprule
Model & Price & Delta & Gamma & Vega & PV01 & Hedge P\&L \\
\midrule
High-fidelity engine & -- & -- & -- & -- & -- & -- \\
Price-only neural operator & & & & & & \\
Derivative-informed neural operator & & & & & & \\
No-arbitrage derivative-informed operator & & & & & & \\
\bottomrule
\end{tabular}
\end{center}

\bibliographystyle{plainnat}
\bibliography{difno_finance_Aplus_polished_refs}

\end{document}